\documentclass[%
reprint,
superscriptaddress,
frontmatterverbose, 
preprintnumbers,
 amsmath,amssymb,
 aps,
 pra,
onecolumn,
notitlepage
]{revtex4-1}

\usepackage{graphicx}
\usepackage{dcolumn}
\usepackage{bm}
\usepackage{hyperref}
\usepackage{subfigure}
\usepackage{comment}

\usepackage{bibentry}

\hypersetup{
colorlinks=true,
linkcolor=blue,
filecolor=blue,
citecolor=blue,  
urlcolor=black,
}

\usepackage{amsthm}
\theoremstyle{plain}
\newtheorem{thm}{Theorem}
\newtheorem{lem}[thm]{Lemma}

\newtheorem{cor}[thm]{Corollary}

\theoremstyle{definition}

\newtheorem{conj}{Conjecture}

\newtheorem*{rem}{Remark}

\newcommand{\ket}[1]{|#1\rangle}
\newcommand{\bra}[1]{\langle #1|}
\newcommand{\bracket}[2]{\langle #1|#2\rangle}
\newcommand{\ketbra}[2]{|#1\rangle\langle #2|}

\newcommand{\ee}{\mathcal{E}}

\newcommand{\eq}[1]{(\hyperref[eq:#1]{\ref*{eq:#1}})}

\renewcommand{\sec}[1]{\hyperref[sec:#1]{Section~\ref*{sec:#1}}}
\newcommand{\thrm}[1]{\hyperref[thm:#1]{Theorem~\ref*{thm:#1}}}
\newcommand{\lemm}[1]{\hyperref[lemm:#1]{Lemma~\ref*{lemm:#1}}}
\newcommand{\prop}[1]{\hyperref[prop:#1]{Proposition~\ref*{prop:#1}}}
\newcommand{\corr}[1]{\hyperref[corr:#1]{Corollary~\ref*{corr:#1}}}
\newcommand{\fig}[1]{\hyperref[fig:#1]{\ref*{fig:#1}}}
\newcommand{\tbl}[1]{\hyperref[fig:#1]{\ref*{tbl:#1}}}
\newcommand{\ba}{\begin{eqnarray}}
\newcommand{\ea}{\end{eqnarray}}

\newcommand{\norm}[1]{\left\lVert#1\right\rVert}

\makeatletter
\newtheorem*{rep@theorem}{\rep@title}
\newcommand{\newreptheorem}[2]{%
\newenvironment{rep#1}[1]{%
 \def\rep@title{#2 \ref{##1}}%
 \begin{rep@theorem}}%
 {\end{rep@theorem}}}
\makeatother

\newreptheorem{theorem}{Theorem}

\newreptheorem{lemma}{Lemma}

\newcommand{\change}[1]{{#1}}

\begin{document}

\preprint{MIT-CTP/4902}

\title{Diagonal quantum discord}

\author{Zi-Wen Liu}
\email{zwliu@mit.edu}
\affiliation{Center for Theoretical Physics, Massachusetts Institute of Technology, Cambridge, Massachusetts 02139, USA}
\affiliation{Research Laboratory of Electronics, Massachusetts Institute of Technology, Cambridge, Massachusetts 02139, USA}
\affiliation{Department of Physics, Massachusetts Institute of Technology, Cambridge, Massachusetts 02139, USA}
\affiliation{Perimeter Institute for Theoretical Physics, Waterloo, Ontario N2L 2Y5, Canada}
\author{Ryuji Takagi}
\email{rtakagi@mit.edu}
\affiliation{Center for Theoretical Physics, Massachusetts Institute of Technology, Cambridge, Massachusetts 02139, USA}
\affiliation{Department of Physics, Massachusetts Institute of Technology, Cambridge, Massachusetts 02139, USA}
\author{Seth Lloyd}
\email{slloyd@mit.edu}
\affiliation{Research Laboratory of Electronics, Massachusetts Institute of Technology, Cambridge, Massachusetts 02139, USA}
\affiliation{Department of Physics, Massachusetts Institute of Technology, Cambridge, Massachusetts 02139, USA}
\affiliation{Department of Mechanical Engineering, Massachusetts Institute of Technology, Cambridge, Massachusetts 02139, USA}


\date{\today}

\begin{abstract}

Quantum discord measures quantum correlation by comparing the quantum mutual
information with the maximal amount of mutual information accessible to
a quantum measurement.    This paper analyzes the properties of diagonal
discord, a simplified version of discord that compares quantum mutual information with the mutual
information revealed by a measurement that correspond to the
eigenstates of the local density matrices.   In contrast to the optimized
discord, diagonal discord is easily computable; it also finds connections to thermodynamics and resource theory.
Here we further show that, for the generic case of non-degenerate local
density matrices, diagonal discord exhibits desirable properties as a preferable discord measure. 
 We employ the
theory of resource destroying maps [Liu/Hu/Lloyd, PRL 118, 060502 (2017)] to prove that diagonal discord is
monotonically nonincreasing under the operation of local 
discord nongenerating qudit channels, $d>2$, and provide
numerical evidence that such monotonicity holds for qubit channels as well.
We also show that it is continuous, and derive
a Fannes-like continuity bound.  
Our results hold for a variety of simple discord measures generalized from
diagonal discord.

\end{abstract}

\maketitle


\section{Introduction}
Quantum discord measures a very general form of non-classical correlation,
which can be present in quantum systems even in the absence of entanglement. Since the first expositions of this concept more than a decade ago \cite{zurek,hv}, a substantial amount of research effort has been devoted to understanding the mathematical properties and physical meanings of discord and similar quantities. 
Comprehensive surveys of the properties of discord can be found in \cite{RevModPhys.84.1655,adesso}, 
and references \cite{doi:10.1142/S123016121440006X,Brodutch2017,2016arXiv161101959A,Vedral2017} provide recent perspectives on the field.

The study of discord presents many challenges and open questions.  
One major difficulty with discord-like quantities is that they are
hard to compute or analyze.
The canonical version of discord is defined to be the difference between quantum mutual information (total correlation) and the maximum amount of correlation that is locally accessible (classical correlation), which involves optimization over all possible local measurements. Such optimization renders the problem of studying discord and its variants (such as quantum deficit \cite{Horodecki2005local}, geometric discord \cite{Dakic2010}) very difficult.  In general, computing these optimized quantities is NP-complete \cite{npc}, and the analytic formulas are only known for very limited cases \cite{Fanchini2010,PhysRevA.84.042313,PhysRevA.86.032110}.

Diagonal discord is a natural simplification of discord, in which one looks at the mutual information revealed by a measurement in the optimal eigenbasis (unique in the absence of degeneracy) of the reduced density matrix of the subsystem under study \cite{heat}. That is, we allow the local density
matrix to `choose' to define mutual information by the locally
minimally disturbing measurement. Because such measurement does not disturb the local states, diagonal discord truly represents the property of `correlation'. 
Note that diagonal discord needs to be distinguished from basis-dependent discord \cite{RevModPhys.84.1655}, which is given by a prefixed local measurement and hence can be studied with tools from coherence theory \cite{PhysRevX.6.041028}.
By definition, diagonal discord is an upper bound for discord as originally defined, \change{and is a faithful discord measure, meaning that it takes zero only for states with zero discord, or equivalently the classical-quantum states}.
Different entropic measures of discord (the optimized discord and deficit \cite{RevModPhys.84.1655,zurekmd}) coincide with diagonal discord when the optimization
procedure leads to measurements with respect to a local eigenbasis.
We note that quantities defined by a similar local measurement strategy have been considered before: the so-called measurement-induced disturbance \cite{mid} and nonlocality \cite{min} are close variants of diagonal discord defined by local eigenbases as well,  but crucially they are not faithful one-way discord measures, as opposed to diagonal discord.
Diagonal discord has been shown to play key roles in thermodynamic scenarios, such as energy transport \cite{heat}, work extraction \cite{dddemon}, temperature estimation \cite{sone2018quantifying}, \change{and local parameter estimation \cite{sone2018nonclassical}}.
In contrast to optimized discord-type quantities, diagonal
discord is in general efficiently computable.  
A similar case is the entanglement negativity \cite{neg}, which, as 
a computable measure of entanglement, greatly simplifies the study of
entanglement in a wide variety of scenarios. Furthermore, diagonal discord naturally emerges from the theory of resource destroying maps \cite{rdm}, a recent general framework for analyzing quantum resource theories.
We believe that the study of diagonal discord
may forge new links between discord and resource theory.

Because diagonal discord is defined
without optimization over local measurements, several 
of its important mathematical properties must be verified.    
First, monotonicity (nonincreasing property) under operations that are considered free is
a defining feature of resource measures; identifying such monotones
is a central theme of resource theory. A curious property of discord is that it can even be created by some local operations \cite{local,Hu2012a}.  
It is unclear whether some other local operation can increase diagonal discord. Note that the monotonicity under all nongenerating operations is arguably an overly strong requirement \cite{adesso}, which automatically implies monotonicity for all theories with less free operations.
Second, continuity is also a desirable feature \cite{criteria,adesso}, which indicates that the measure does not see a sudden jump under small perturbations. 
From examples given in \cite{criteria,wpm}, where the local states are both maximally mixed qubits, we know that diagonal discord can generally be discontinuous at degeneracies. However, the continuity properties otherwise remains unexplored.  These two unclear features represent the most important concerns of restricting to local eigenbases. 

The purpose of this paper is to address the above concerns. We first find that, rather surprisingly, diagonal discord exhibits good monotonicity properties under local discord nongenerating operations. 
The discord cannot be generated under local operation if and only if the local operation is commutativity-preserving \cite{Hu2012a}. 
We show that local isotropic channels, a subset of commutativity-preserving maps, commute with the canonical discord destroying map, which implies that diagonal discord is monotone under them by \cite{rdm}. By the classification of commutativity-preserving operations \cite{local,Hu2012a,Guo2013}, we conclude that monotonicity holds for all local commutativity-preserving operations except for unital qubit channels that are not isotropic. However, numerical studies imply that monotonicity holds for these channels as well. Then, we prove that, when the local density operator is nondegenerate, diagonal discord is continuous. We derive a Fannes-type continuity bound, which diverges as the minimum gap between eigenvalues tends to zero as expected. 
\change{At last, we explicitly compare diagonal discord with the optimized discord for a large number of randomly sampled symmetric two-qubit $X$-states, which are expected to reveal the generic behaviors of bipartite states, to better understand the simplification of measurement strategy.  We find that, for a significant proportion of states, the optimal measurement that induces discord is given by a local eigenbasis, or equivalently, diagonal discord matches the optimized discord. 
In other cases, the value of diagonal discord could be significantly greater than discord. However, it should be emphasized that diagonal discord should just be seen as a different way of measuring the same type of resource, which correspond to different operational and physical meanings.  In this sense, it is not very meaningful to directly compare the values of diagonal discord and optimized discord.}

\section{Diagonal discord}
Here, we define the notion of diagonal discord more formally. Without loss of generality, we mainly study the one-sided discord of a bipartite state $\rho_{AB}$, where the local measurements are made on subsystem $A$. As will be shown later, it is straightforward to generalize the results to two-sided measurements or multipartite cases.

Let $\{\Pi^A_i\equiv|i\rangle_A\langle i|\}$ be a local eigenbasis of $A$, i.e., suppose $\rho_A={\rm tr}_{B}\rho_{AB}$ admits spectral decomposition $\rho_A = \sum_i p_i \Pi_i^A$. 
Note that the eigenbasis is not uniquely determined in the presence of degeneracy in the spectrum.
Define $\pi_A(\rho_{AB})=\sum_i(\Pi^A_i\otimes I_B)\rho_{AB}(\Pi^A_i\otimes I_B)=\sum_i \Pi_i^A\otimes \langle i|\rho_{AB}|i\rangle$, \change{where $\langle i|\rho_{AB}|i\rangle := {\rm tr}_A([\Pi_i^A\otimes I_B]\rho_{AB})$}. 
This describes the local measurement in some eigenbasis $\{\Pi_i^A\}$.

Diagonal discord of $\rho_{AB}$ as measured by $A$, denoted as $\bar{D}_A(\rho_{AB})$, quantifies the reduction in mutual information induced by $\pi_A$. Since $\pi_A$ does not perturb $\rho_A$, $\bar{D}_A$ equals the increase in the global entropy. So diagonal discord represents a unified simplification of discord and deficit.
Formally, 
\begin{eqnarray}
\bar{D}_A(\rho_{AB})&:=&I(\rho_{AB})-\change{\max_{\pi_A}} I(\pi_A(\rho_{AB}))\label{discord}\\
&=&\change{\min_{\pi_A}} S(\pi_A(\rho_{AB}))-S(\rho_{AB}).\label{ssdd}\label{deficit}
\end{eqnarray}
where the optimization is taken over the local eigenbases spanning the possibly degenerate subspace.
\change{(To avoid issues concerning the existence of optimal measurements, we use minimum instead of infimum, following e.g.~\cite{hv,2010arXiv1004.2082M,PhysRevA.88.014302}.)}
If the optimization is instead taken over all local measurements, the first line Eq.~(\ref{discord}) would reduce to discord \change{(note that the two original definitions of discord differ slightly in the local measurements allowed: Ollivier and Zurek \cite{zurek} used von Neumann measurements, while Henderson and Vedral \cite{hv} used POVMs)}, and the second line Eq.~(\ref{deficit}) would reduce to deficit, which are inequivalent in general.

It is crucial that diagonal discord can also take the form of relative entropy.  First notice that
\begin{eqnarray}\label{pitr}
S(\pi_A(\rho_{AB}))&=&-{\rm tr}(\pi_A(\rho_{AB})\log \pi_A(\rho_{AB}))
=-{\rm tr}\left(\sum_i (\Pi^A_i\otimes I_B)\rho_{AB}(\Pi^A_i\otimes I_B)\log \pi_A(\rho_{AB})\right)\\
&=&-{\rm tr}\left(\sum_i (\Pi^A_i\otimes I_B)\rho_{AB}\log \pi_A(\rho_{AB})(\Pi^A_i\otimes I_B)\right)\\
&=&-{\rm tr}\left(\sum_i (\Pi^A_i\otimes I_B)^2\rho_{AB}\log \pi_A(\rho_{AB})\right)\\
&=&-{\rm tr}\left(\sum_i (\Pi^A_i\otimes I_B)\rho_{AB}\log \pi_A(\rho_{AB})\right)\\
&=&-{\rm tr}\left(\rho_{AB}\log \pi_A(\rho_{AB})\right),
\end{eqnarray}
where the second line follows from the fact that each $(\Pi^A_i\otimes I_B)$ commutes with $\pi_A(\rho_{AB})$, the third line follows from the cyclic property of trace, the fourth line follows from the idempotence of $\Pi^A_i\otimes I_B$, and the fifth line follows from the completeness relation $\sum_i(\Pi^A_i\otimes I_B)=I$. Therefore, by Eq.~(\ref{deficit}),
\begin{equation}
    \bar{D}_A(\rho_{AB}) = \change{\min_{\pi_A}} S(\pi_A(\rho_{AB})) - S(\rho_{AB}) =\change{\min_{\pi_A}} {\rm tr}\{\rho_{AB}\left(\log\rho_{AB}-\log{\pi_A(\rho_{AB})}\right)\} = \change{\min_{\pi_A}}S(\rho_{AB}\|\pi_A(\rho_{AB})).  \label{reldd}
\end{equation}
That is, diagonal discord of $\rho_{AB}$ equals the relative entropy to $\pi_A(\rho_{AB})$, minimized over eigenbases $\pi_A$ in the presence of degeneracies.
From the above relation, it can be seen that diagonal discord indeed obeys the faithfulness condition, that is, it only vanishes for classical-quantum states \change{(as the optimized discord \cite{zurek,datta,fer})}, the fixed points of $\pi_A$.
Note that, in general, optimization is still needed within degenerate subspaces.   
It can be seen that, as long as the degenerate subspace is small, diagonal discord can be efficiently computed. In this paper, we are mostly concerned with the nondegenerate case, where $\pi_A$ is unique.
\change{Note that the (one-sided version of) measurement-induced disturbance \cite{mid} and measurement-induced nonlocality \cite{min,xi2012} are not faithful due to the absence of the above minimization within the degenerate subspace \cite{Girolami2011}, which is a crucial difference from the diagonal discord.}

\section{Structure of $\pi$ theory and monotonicity}

In this section, we investigate the monotonicity property of diagonal discord.  The main idea is to employ the monotonicity theorem of commuting free operations, which comes from the theory of resource destroying maps \cite{rdm}.  To do so, we need to analyze whether the discord-free operations commutes with certain resource destroying map for discord.  The cases of qubits and higher dimensions turn out to be quite different and are discussed separately.

\subsection{Resource destroying maps and the monotonicity theorem}

Before going into details, we first briefly review the theory of resource destroying maps \cite{rdm}.
A map $\lambda$ is called resource destroying map if it maps all non-free states to free ones, and does nothing on free states.
It allows to classify quantum channels (completely-positive trace preserving maps) into some classes depending on the condition that the channel satisfies. 
Let $\mathcal{E}$ be a channel. The nongenerating condition $\lambda\circ\mathcal{E}\circ\lambda = \mathcal{E}\circ\lambda$ gives the maximal nontrivial set of free operations; a little thought will convince one that it captures the resource nongenerating property of $\ee$. 
It turns out to be useful to consider a stronger condition $\lambda\circ\mathcal{E} = \mathcal{E}\circ\lambda$, which we call the commuting condition. 
Let $\bar{X}(\lambda)$ and $X(\lambda)$ denote the sets of operations satisfying the nongenerating condition and commuting condition respectively. 
Note that any channel in \change{$X(\lambda)$} is also in \change{$\bar{X}(\lambda)$} but the converse is not true in general.
This classification is very useful when the monotonicity of resource measures is concerned. One can easily show the following:
\begin{thm}[\cite{rdm}]
Let $\lambda$ be a resource destroying map.  Then the distance-based resource measure $\delta(\rho,\lambda(\rho))$, where $\delta$ is any distance measure satisfying the data processing inequality, is nonincreasing under $X(\lambda)$ (channels that satisfy the commuting condition).  
\label{thm:mono}
\end{thm}

\subsection{Theory of discord destroying map $\pi$ and the monotonicity of diagonal discord}

When it comes to the discord where free states are classical-quantum states, $\pi_A$ (defined in the last section) is a natural discord destroying map.   We stress that $\pi_A$ is not a quantum channel. (Since the set of classical-quantum states is nonconvex, any discord destroying map is nonlinear \cite{rdm}.) By Eq.~(\ref{reldd}), taking $\pi_A$ as the resource destroying map and the relative entropy as the distance measure gives us diagonal discord. 

\subsubsection{Classification of discord nongenerating local operations}
We consider operations acting on subsystem $A$, \change{the dimension of whose Hilbert space is denoted by $d_A$}.  The largest possible set of free operations is the set of discord nongenerating channels of the form $\ee \otimes I_B \in \bar{X}(\pi_A)$. We call such $\ee$ local discord nongenerating channels and write the set of local discord nongenerating channels as $\bar{X}_A(\pi_A)$.   Recall the definitions of the following classes of channels:
\begin{itemize}
    \item Mixed-unitary channels ({\bf MU}): $\mathcal{E}^{\mathrm{MU}}(\rho) = \sum_\mu p_\mu U_\mu \rho U_\mu^\dagger$, where $\{p_{\mu}\}$ is a probability distribution ($p_{\mu}\in[0,1],\sum_\mu p_\mu=1$) and $U_{\mu}$ are unitary channels.
    \item Isotropic channels ({\bf ISO}): $\mathcal{E}^{\mathrm{ISO}}(\rho) = (1-\gamma)W(\rho) + \gamma I/d$,
where \change{$\gamma\in[0,1]$, $d$ denotes the dimension of the Hilbert space so $I/d$ is the maximally mixed state},  and $W$ is either unitary or antiunitary.
\item Semiclassical channels ({\bf SC}): channels whose outputs are diagonal with a certain preferred basis.
\end{itemize}

It has been shown that $\ee \in \bar{X}_A(\pi_A)$ if and only if $\ee$ is commutativity-preserving \cite{Hu2012a}, and the set of commutativity-preserving channels consist of unital channels for qubits ($d_A=2$) and isotropic channels for qudits ($d_A>2$) \cite{local,Hu2012a,Guo2013}, in addition to all semiclassical channels (which always destroy discord). Note that, for qubits, the set of unital channels is equivalent to the set of mixed-unitary channels \cite{watrous}.
The structure of $\bar{X}_A(\pi_A)$ is depicted in Fig.~\fig{cp}.
\begin{figure} 
\centering
\includegraphics[width=0.5\columnwidth]{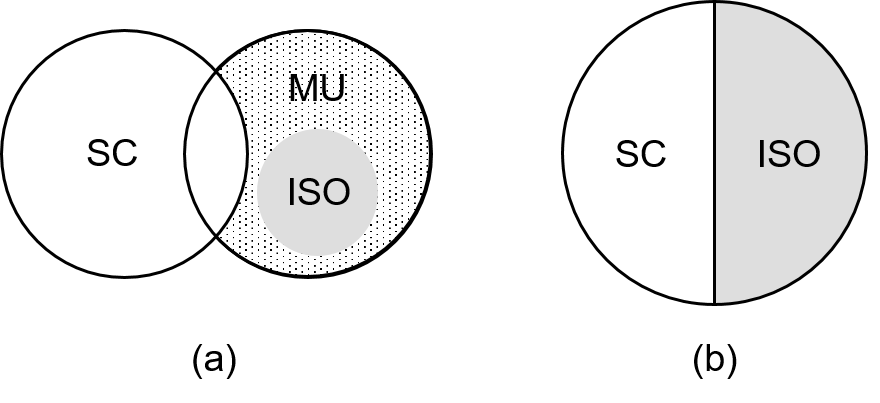}
\caption{\label{fig:cp} Structure of commutativity-preserving channels ($\bar{X}_A(\pi_A)$). Note that we define $\mathbf{SC}$ to exclude the completely depolarizing channel. (a) Qubits: note that $\mathbf{ISO}\subsetneq\mathbf{MU}$, $\mathbf{SC}\cap\mathbf{MU}$ includes e.g.~completely dephasing channels; (b) Qudits with $d>2$. Grey area ({\bf ISO}): in $X_A(\pi_A)$; Dotted area ($\mathbf{MU}\backslash(\mathbf{ISO}\cup\mathbf{SC})$): numerical evidence of being in $X_A(\pi_A)$; White area ({\bf SC}): not in $X_A(\pi_A)$.}
\end{figure}
Note that the completely depolarizing channel, which maps everything to the maximally mixed state, is in $\mathbf{SC}$ as well as $\mathbf{ISO}$. We exclude the completely depolarizing channel from $\mathbf{SC}$ so that $\mathbf{SC}\cap\mathbf{ISO}=\emptyset$.
For simplicity,  $\rho_A$ is assumed to be nondegenerate.
We are interested in the monotonicity of $\bar{D}_A$ under $\ee \otimes I$ for $\ee \in \bar{X}_A(\pi_A)$.

\subsubsection{$\pi$-commuting local operations and monotonicity}
In general, identifying operations under which some measure behaves as a monotone is a highly nontrivial task.
Due to the relative entropy form Eq.~(\ref{reldd}), the above monotonicity theorem (Theorem \ref{thm:mono}) can be applied to diagonal discord: $\bar{D}_A$ is monotonically nonincreasing under $\ee\otimes I_B \in X(\pi_A)$. Let ${X}_A(\pi_A)$ denote the set of local operations on $A$ that commute with $\pi_A$ together with identity operation on $B$.
We now identify operations that belong to ${X}_A(\pi_A)$. First, it is known that \change{$\mathbf{SC} \cap {X}_A(\pi_A) = \emptyset$} \cite{rdm}.
Ref.~\cite{rdm} also showed that unitary-isotropic channels are in ${X}_A(\pi_A)$. 
We analyze the remaining cases for qubits and qudits separately, since they exhibit different structures in the theory of $\pi$.

\paragraph{Qubit ($d_A=2$).}
For the qubit case, we derive an explicit local condition that determines if a local mixed-unitary channel is in $X_A(\pi_A)$:
\begin{lem}\label{qubitcondition}
Consider the qubit mixed-unitary channel $\mathcal{E}^{\mathrm{MU}}(\rho) = \sum_\mu p_\mu U_\mu \rho U_\mu^\dagger$. Let $\{\ket{\psi},\ket{\bar\psi}\}$ be some orthonormal basis, and $\{\ket{\eta_+},\ket{\eta_-}\}$ be the common eigenbasis of $\mathcal{E}^{\mathrm{MU}}(\ketbra{\psi}{\psi})$ and $\mathcal{E}^{\mathrm{MU}}(\ketbra{\bar{\psi}}{\bar{\psi}})$.  
Then $\mathcal{E}^{\mathrm{MU}}\in {X}_A(\pi_A)$ if and only if
\begin{equation}
 \sum_\mu p_\mu \bra{\eta_l}U_\mu\ket{\psi}\bra{\bar{\psi}}U_\mu^{\dagger}\ket{\eta_l}=0
 \label{eq:condition0}
\end{equation}
for any choice of basis $\{\ket{\psi},\ket{\bar{\psi}}\}$ and $l=+,-$.
\end{lem}
\begin{proof}

Since qubit mixed-unitary channels are commutativity-preserving, 
all input states with the same eigenbasis share a common output eigenbasis.
Let \{$\ket{0},\ket{1}$\} be the eigenbasis of $\rho_A$.
Denote the common eigenbasis of $\mathcal{E}^{\mathrm{MU}}(\ketbra{0}{0})$ and 
$\mathcal{E}^{\mathrm{MU}}(\ketbra{1}{1})$ as $\{\ket{\eta_+}, \ket{\eta_-}\}$, and the corresponding eigenvalues $\eta^i_+,\eta^i_-$ for $i=0,1$, that is, $\mathcal{E}^{\mathrm{MU}}(\ketbra{i}{i})=\sum_{l=+,-} \eta^i_l \ketbra{\eta_l}{\eta_l}$. Note that, since $\sum_i\mathcal{E}^{\mathrm{MU}}(\ketbra{i}{i})=I$, $\eta^0_l =1-\eta^1_l$ for $l=+,-$. By linearity, $\mathcal{E}^{\mathrm{MU}}(\rho_A)$ admits a spectral decomposition in the basis $\{\ket{\eta_+}, \ket{\eta_-}\}$.

We first obtain
\begin{eqnarray}
 [\mathcal{E}^{\mathrm{MU}}_A \otimes I_B] \circ \pi_A(\rho_{AB})&=&\sum_{i=0,1} \mathcal{E}^{\mathrm{MU}}(\ketbra{i}{i})\otimes \bra{i}\rho_{AB}\ket{i} \\
 &=&\sum_{i=0,1} \sum_{l=+,-}  \eta^i_l \ketbra{\eta_l}{\eta_l}\otimes \bra{i}\rho_{AB}\ket{i}.\label{epi}
\end{eqnarray}
On the other hand,
\begin{eqnarray}
 \pi_A \circ [\mathcal{E}^{\mathrm{MU}}_A \otimes I_B](\rho_{AB}) &=& \sum_{l=+,-} \ket{\eta_l}\bra{\eta_l}\otimes \bra{\eta_l}[\mathcal{E}^{\mathrm{MU}}_A \otimes I_B](\rho_{AB})\ket{\eta_l}\\
 &=& \sum_\mu \sum_{l=+,-} p_\mu \ket{\eta_l}\bra{\eta_l}\otimes \bra{\eta_l}(U_\mu\otimes I) \rho_{AB} (U_\mu^{\dagger}\otimes I)\ket{\eta_l}\\
 &=&\sum_{i,j=0,1}\sum_\mu \sum_{l=+,-} p_\mu \bra{\eta_l}U_\mu\ket{i}\bra{j}U_\mu^\dagger\ket{\eta_l}\ketbra{\eta_l}{\eta_l} \otimes \bra{i}\rho_{AB}\ket{j}\\
 &=&\sum_{i=0,1} \sum_{l=+,-} \eta^i_l \ketbra{\eta_l}{\eta_l} \otimes \bra{i}\rho_{AB}\ket{i}
 +\sum_{i\neq j}\sum_\mu\sum_{l=+,-} p_\mu \bra{\eta_l}U_\mu\ket{i}\bra{j}U_\mu^{\dagger}\ket{\eta_l} \ketbra{\eta_l}{\eta_l} \otimes \bra{i} \rho_{AB}\ket{j},\label{pie}
\end{eqnarray}
where the first line follows from the spectral decomposition of $\mathcal{E}^{\mathrm{MU}}(\rho_A)$.
Therefore, the two sides of the commuting condition Eqs.~(\ref{epi}) and (\ref{pie}) coincide  if and only if
\begin{equation}
\change{M_l} + \change{M_l}^\dagger = 0 \label{eq:condition_0}
\end{equation}
for $l=+,-$, where
\begin{equation}
\change{M_l} = \sum_\mu p_\mu \bra{\eta_l}U_\mu\ket{0}\bra{1}U_\mu^{\dagger}  \ket{\eta_l} \bra{0} \rho_{AB}\ket{1}
\end{equation}
is a matrix defined on $B$.
In other words, $\change{M_l}$ is a skew-Hermitian matrix: it has zero or pure imaginary diagonal entries. Since the diagonals of $\bra{0}\rho_{AB}\ket{1}$ can be real or imaginary depending on $\rho_{AB}$ (for example, for $\rho_{AB}=\change{\frac12}[\ket{0+}+\ket{1-}][\bra{0{+}}+\bra{1{-}}]$, \change{$\bra{00}\rho_{AB}\ket{10}=1/4$ and $\bra{01}\rho_{AB}\ket{11}=-1/4$, but for $\rho_{AB}=\change{\frac12}[\ket{0+}+i\ket{1-}][\bra{0{+}}-i\bra{1{-}}]$, $\bra{00}\rho_{AB}\ket{10}=-i/4$ and $\bra{01}\rho_{AB}\ket{11}=i/4$}), $\change{M_l}$ can only be the zero matrix so that the skew-Hermitian condition holds for arbitrary $\rho_{AB}$. 
Furthermore, notice that the eigenbasis can vary arbitrarily depending on $\rho_{AB}$, so Eq.~\eq{condition_0} must hold for any basis. 
Therefore, Eq.~\eq{condition_0} is reduced to the following final condition. Let $\{\ket{\psi},\ket{\bar{\psi}}\}$ be some orthonormal basis, and $\{\ket{\eta_+^{\psi}},\ket{\eta_-^{\psi}}\}$ be the common eigenbasis of $\mathcal{E}^{\mathrm{MU}}(\ketbra{\psi}{\psi})$ and 
$\mathcal{E}^{\mathrm{MU}}(\ketbra{\bar{\psi}}{\bar{\psi}})$.  Then $\mathcal{E}^{\mathrm{MU}}_A \otimes I_B$ and $ \pi_A$ commute if and only if
\begin{equation}
 \sum_\mu p_\mu \bra{\eta_l^{\psi}}U_\mu\ket{\psi}\bra{\bar{\psi}}U_\mu^{\dagger}\ket{\eta_l^{\psi}}=0
 \label{eq:condition}
\end{equation}
for any choice of basis $\{\ket{\psi},\ket{\bar{\psi}}\}$ and $l=+,-$.

\end{proof}

By explicitly using Lemma \ref{qubitcondition}, we can show that all isotropic channels are in ${X}_A(\pi_A)$:
\begin{thm}\label{qubit_iso}
For $d_A=2$, ${\bf ISO} \subset X_A(\pi_A)$.
\end{thm}
\proof{
Here we show that ${\bf ISO} \subset X_A(\pi_A)$ for qubits by directly employing the condition introduced in Lemma \ref{qubitcondition}.

Unitary-isotropic channels are already shown to be in $X_A(\pi_A)$ \cite{rdm}.
One can confirm that unitary-isotropic channels indeed satisfy the condition as follows.
Consider a qubit unitary-isotropic channel $u(\rho)=(1-\gamma)U\rho U^{\dagger} +\gamma I/2=(1-\gamma)U\rho U^{\dagger} +\frac{\gamma}{4}(X\rho X+Y\rho Y+Z\rho Z + I\rho I)$, where $X,Y,Z$ are defined to be Pauli matrices in the basis $\{\ket{\psi},\ket{\bar{\psi}}\}$. 
Here the basis $\{\ket{\psi},\ket{\bar{\psi}}\}$ can be arbitrarily chosen since the identity operator can be decomposed as the uniform Pauli twirling \cite{nc,wilde} in any basis.
It is clear that $\ket{\eta_{+,-}}=\{U\ket{\psi},U\ket{\bar{\psi}}\}$. One can verify that $U\ket{\psi}$ satisfies the condition as follows. The first term (unitary component) gives
\begin{equation}
 (1-\gamma)\bra{\psi}U^{\dagger}U\ket{\bar{\psi}}\bra{\psi}U^{\dagger}U\ket{\psi}=0.\label{u0}
\end{equation}
For the Pauli components, we obtain
\begin{eqnarray}
 &&\bra{\psi}U^{\dagger}X \ketbra{\bar{\psi}}{\psi}X U\ket{\psi} + 
 \bra{\psi}U^{\dagger}Y \ketbra{\bar{\psi}}{\psi}Y U\ket{\psi} + 
 \bra{\psi}U^{\dagger}Z \ketbra{\bar{\psi}}{\psi}Z U\ket{\psi} + \bra{\psi}U^{\dagger} \ketbra{\bar{\psi}}{\psi} U\ket{\psi}\nonumber \\
 &=&\bra{\psi}U^{\dagger} \ketbra{{\psi}}{\bar\psi} U\ket{\psi} - 
 \bra{\psi}U^{\dagger} \ketbra{{\psi}}{\bar\psi} U\ket{\psi} - 
 \bra{\psi}U^{\dagger} \ketbra{\bar{\psi}}{\psi} U\ket{\psi} + \bra{\psi}U^{\dagger} \ketbra{\bar{\psi}}{\psi} U\ket{\psi} = 0,\label{twirl0}
\end{eqnarray}
by plugging in the Pauli matrices.
It can be seen that the terms corresponding to $X,Y$ and $Z,I$ respectively cancel each other. The condition holds for $U\ket{\bar{\psi}}$ as well.
So we conclude that $u\in X_A(\pi_A)$.

 We now show that any antiunitary-isotropic channel $\bar{u}(\rho)=(1-\gamma) U\rho^T U^{\dagger} +\gamma I/2$ also satisfies the condition.
Let \{$V\ket{\psi},V\ket{\bar{\psi}}$\} be the basis with respect to which the transpose is taken, where $V$ is unitary.  
Notice that transpose operation can be written as $\rho^T=\frac{1}{2}(\rho+X_V\rho X_V - Y_V\rho Y_V +Z_V\rho Z_V)$, where $X_V=VXV^{\dagger},Y_V=VYV^{\dagger},Z_V=VZV^{\dagger}$ are Pauli matrices in the transposition basis. So
\begin{eqnarray}
\bar{u}(\rho)&=&\frac{1-\gamma}{2}U(\rho+X_V\rho X_V - Y_V\rho Y_V +Z_V\rho Z_V)U^{\dagger}+\gamma \frac{I}{2}\\
&=&(2-\gamma)\frac{I}{2} -(1-\gamma)UY_V\rho Y_VU^{\dagger}.
\end{eqnarray}
 
We are now ready to examine whether $\bar{u}(\rho)$ satisfies the condition.
Due to Eq.~(\ref{twirl0}), the first term gives zero.  Notice that the new eigenbasis is  $\ket{\eta_{+,-}}=\{U'\ket{\psi},U'\ket{\bar{\psi}}\}$, where $U'= UY_V$ is unitary. So the second term also gives zero due to Eq.~(\ref{u0}). So $\bar{u}\in X_A(\pi_A)$.
\qed
}

Therefore, \change{combining with the fact that local semiclassical channels always output classical-quantum states (with zero discord and diagonal discord) by definition,} we obtain the following result for qubits:
\begin{cor}
For $d_A=2$, diagonal discord is monotonically nonincreasing under ${\bf SC}\cup {\bf ISO}$.
\end{cor}

However, the condition in Lemma \ref{qubitcondition} does not hold in general, which implies that ${X}_A(\pi_A)\subsetneq\mathbf{MU}$ for qubits. For instance, consider $\ee(\cdot)=\frac{1}{3}I(\cdot)I + \frac{2}{3} H (\cdot)H$ where $H$ is the Hadamard transformation in the computational basis $\{\ket{0},\ket{1}\}$. Straightforward calculation gives $\ket{\eta_+}=\frac{1}{\sqrt{N}}(\ket{0}+\frac{-1+\sqrt{5}}{2}\ket{1})$ and $\ket{\eta_-}=\frac{1}{\sqrt{N}}(\frac{-1+\sqrt{5}}{2}\ket{0}-\ket{1})$ where $N$ is the normalization factor. Then
\begin{eqnarray*}
\sum_i p_i \bra{\eta_+}U_i\ket{1}\bra{0}U_i^{\dagger}\ket{\eta_+}
=\frac{1}{3}\bracket{\eta_+}{0}\bracket{1}{\eta_+}+\frac{2}{3}\bracket{\eta_+}{+}\bracket{-}{\eta_+}=\frac{1}{3N}(\sqrt{5}-1) \neq 0.
\end{eqnarray*} 
So this probabilistic Hadamard is not in ${X}_A(\pi_A)$.
We conjecture (which is not important for our current purpose) that ${\bf ISO} = X_A(\pi_A)$. That is, qubit mixed-unitary channels that are not isotropic all fail to satisfy the condition.

For qubit channels that live in ${\bf MU}\setminus X_A(\pi_A)$, the current idea for proving monotonicity do not apply.
However, we provide numerical results which strongly indicate that diagonal discord is monotone under such channels as well. Fig.~\fig{monotonicity} displays the comparison between diagonal discord before and after the action of several typical non-isotropic mixed-unitary channels, for a large number of randomly generated input states. 
It can be seen that all data points reside on the nonincreasing side. All other channels that we have analyzed exhibit similar behaviors. 
We put this as a conjecture at the moment:
\begin{conj}
For $d_A= 2$, diagonal discord is monotonically nonincreasing under any local discord nongenerating channel.
\end{conj}

\begin{figure*}
\begin{center}
\subfigure[]{
\includegraphics[width=0.32\textwidth]{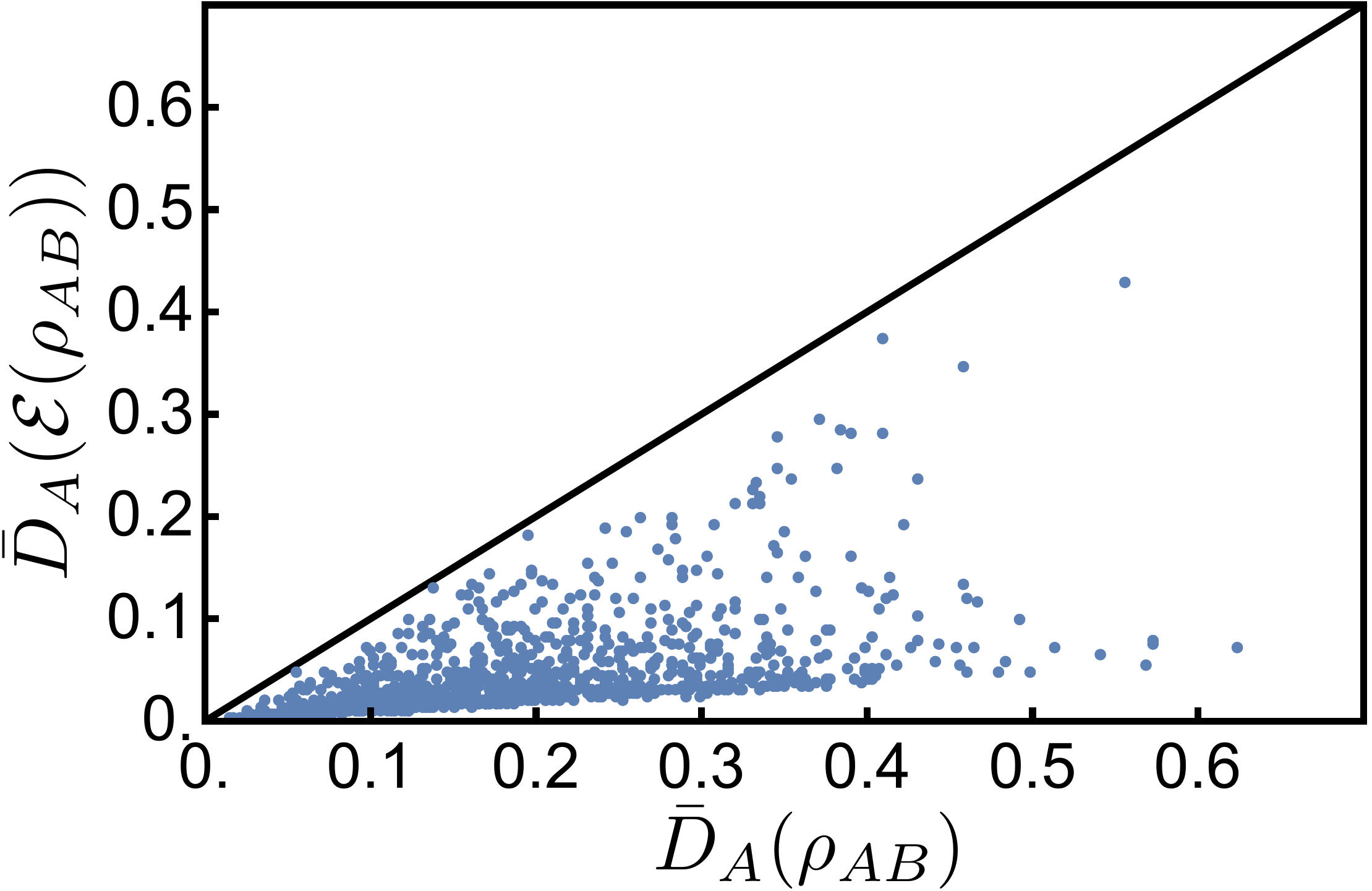}}
\subfigure[]{
\includegraphics[width=0.32\textwidth]{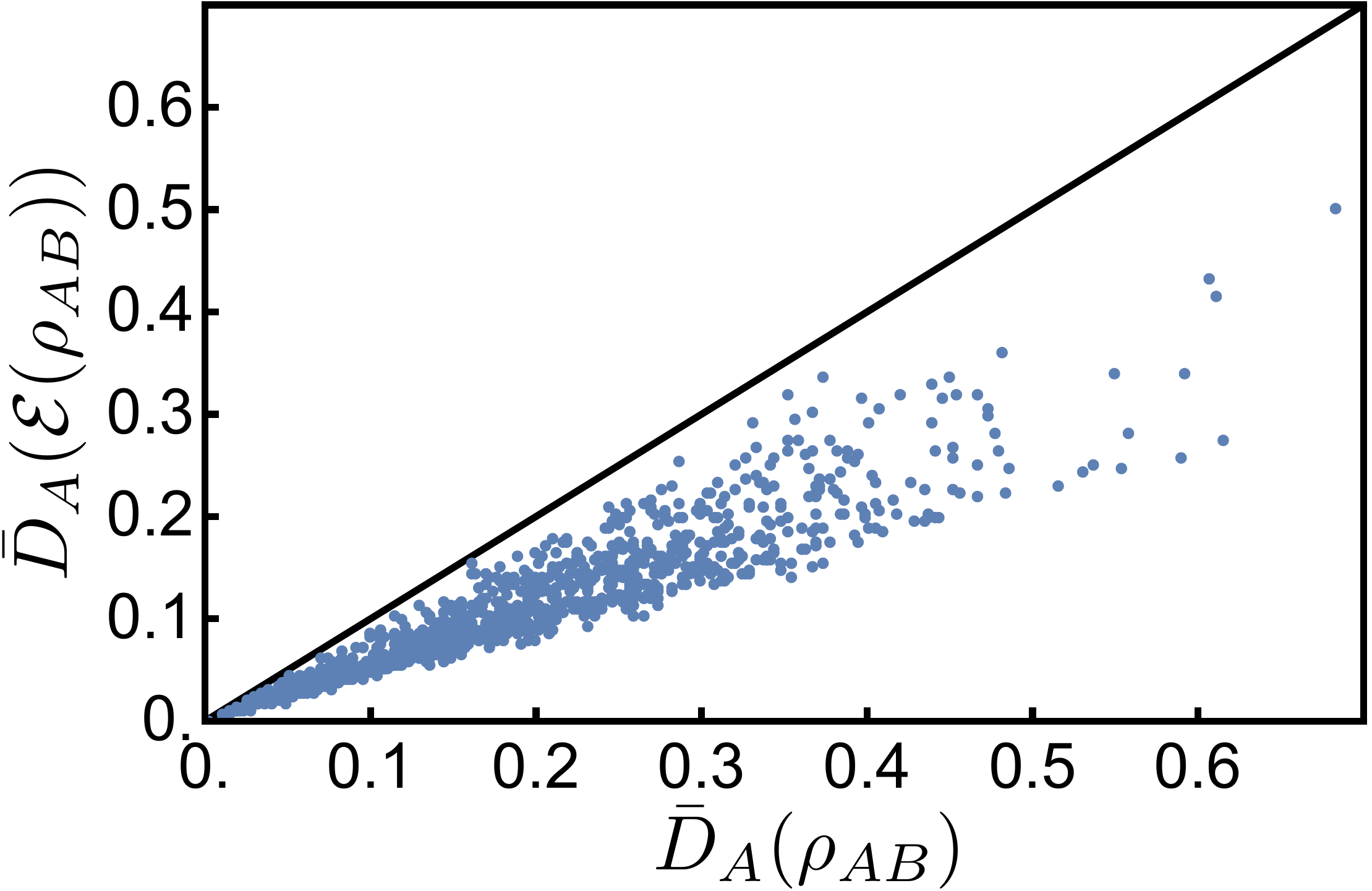}}
\subfigure[]{
\includegraphics[width=0.32\textwidth]{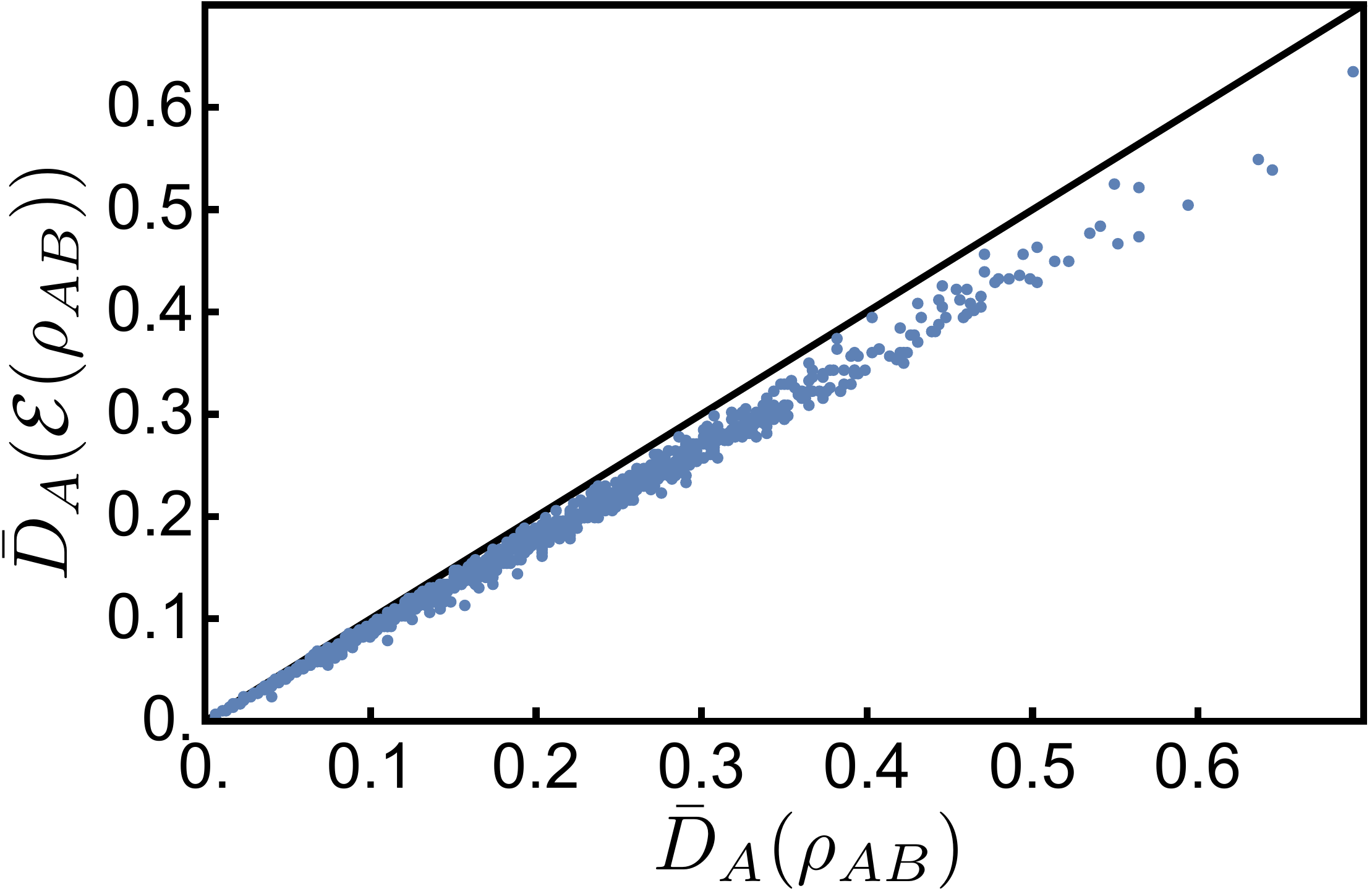}}
\caption{
Comparison of diagonal discord of the input and output states of channels $\mathcal{E} = \ee^{\rm MU}_A\otimes I_B$ such that $\ee^{\rm MU}\in {\bf MU} \setminus X_A(\pi_A)$.
The black line $\bar{D}_A(\ee(\rho_{AB}))=\bar{D}_A(\rho_{AB})$ serves as a baseline for the comparison.
The local mixed unitary channels considered are (a)  $\ee^{\rm MU}(\rho)=\frac{1}{3}\rho +\frac{2}{3}H\rho H$,  (b) $\ee^{\rm MU}(\rho)=\frac{1}{3}\rho +\frac{2}{3}R_n(\pi/2)\rho R_n(\pi/2)^{\dagger}$ where $R_{\bf n}(\pi/2)$ is the $\pi/2$ rotation with respect to the axis ${\bf n}\propto (1,1,1)$, and (c) $\ee^{\rm MU}(\rho)=\frac{1}{6}\rho +\frac{1}{3}R_X(\pi/10)\rho R_X(\pi/10)^{\dagger}+\frac{1}{2}R_Z(\pi/5)\rho R_Z(\pi/5)^{\dagger}$ where $R_X$ and $R_Z$ are rotations with respect to $X$ axis and $Z$ axis respectively. The choice of these channels is arbitrary. The number of samples is set to 1000 for each channel.}
\label{fig:monotonicity}
 \end{center}  
\end{figure*}

\paragraph{Qudit ($d_A>2$).}
The analysis for $d_A>2$ turns out to be simpler.
In fact, it can be shown in general dimensions that ${\bf ISO} \subset X_A(\pi_A)$. The main step of the proof is to explicitly write out the eigenbasis after an antiunitary transformation.
\begin{thm}\label{qudit_iso}
For $d_A\geq 2$, ${\bf ISO} \subset X_A(\pi_A)$. In particular, for $d_A>2$, ${\bf ISO} = X_A(\pi_A)$.
\end{thm}
\begin{proof}
  Here we provide a general proof of ${\bf ISO} \subset X_A(\pi_A)$. Note that Theorem \ref{qubit_iso} for qubit systems is just a special case of this result. For $d_A>2$ we have $\bar{X}_A(\pi_A) = {\bf SC} \cup {\bf ISO}$ \cite{Hu2012a,Guo2013}, so ${\bf ISO} = X_A(\pi_A)$.

Again, recall that unitary-isotropic channels are shown to be in $X_A(\pi_A)$ \cite{rdm}. Here we show that any antiunitary-isotropic channel $\bar{u}(\rho) = (1-\gamma)U\rho^T U^\dagger + \gamma I/d$ is also in $X_A(\pi_A)$ for any $d_A$.
Let \{$\ket{t_i}$\} be the \change{complete orthonormal} basis with respect to which the transposition is taken.
Suppose the input state $\rho_{AB}$ reads $\rho_{AB}=\sum_{ijkl}q_{ijkl}\ketbra{t_i}{t_j}\otimes \ketbra{r_k}{r_l}$, where $i,j$ and $k,l$ are respectively indices of $A$ and $B$, and $\{\ket{r_{k,l}}\}$ denotes some basis of the Hilbert space of $B$. 
Given that the spectral decomposition of $A$ reads $\rho_A = \sum_\alpha\lambda_\alpha\ketbra{\alpha}{\alpha}$, we have
\begin{eqnarray}
\pi_A(\rho_{AB})
&=&\sum_\alpha\sum_{ij}\sum_{kl} q_{ijkl}\bracket{\alpha}{t_i}\bracket{t_j}{\alpha}\ketbra{\alpha}{\alpha}\otimes \ketbra{r_k}{r_l}\\
&=&\sum_\alpha\sum_{ijmn}\sum_{kl} q_{ijkl}\bracket{\alpha}{t_i}\bracket{t_j}{\alpha}\bracket{t_m}{\alpha}\bracket{\alpha}{t_n}\ketbra {t_m}{t_n}\otimes \ketbra{r_k}{r_l},
\end{eqnarray}
and so
\begin{equation}\label{upi}
[\bar{u}_A\otimes I_B] \circ \pi_A(\rho_{AB})
=(1-\gamma)\sum_\alpha\sum_{ijmn}\sum_{kl} q_{ijkl}\bracket{\alpha}{t_i}\bracket{t_j}{\alpha}\bracket{t_m}{\alpha}\bracket{\alpha}{t_n}U\ketbra{t_n}{t_m}U^{\dagger}\otimes \ketbra{r_k}{r_l}
+\gamma\frac{I_A}{d_A}\otimes \rho_B.
\end{equation}
On the other hand,
\begin{equation}\label{ubar}
[\bar{u}_A\otimes I_B](\rho_{AB})
=(1-\gamma)\sum_{ij}\sum_{kl} q_{ijkl} U\ketbra{t_j}{t_i}U^{\dagger}\otimes \ketbra{r_k}{r_l}+\gamma\frac{I_A}{d_A}\otimes \rho_B,
\end{equation}
which involves a partial transpose. 
In order to express the action of $\pi_A$, we need to find the eigenbasis of the reduced density operator 
\begin{equation}
 \mathrm{tr}_B([\bar{u}_A\otimes I_B](\rho_{AB}))=(1-\gamma)\sum_{ij}\sum_{k} q_{ijkk} U\ketbra{t_j}{t_i}U^{\dagger}+\gamma\frac{I_A}{d_A} = (1-\gamma)U\rho_A^T U^\dagger + \gamma\frac{I_A}{d_A}.
\end{equation}
We essentially need to find the eigenbasis of $\rho_A^T$.
Rewrite $\rho_A$ as 
$\rho_A = \sum_{ij}\sum_\alpha \lambda_\alpha\bracket{t_i}{\alpha}\bracket{\alpha}{t_j}\ketbra{t_i}{t_j}$,
that is, $\sum_k q_{ijkk} = \sum_\alpha \lambda_\alpha\bracket{t_i}{\alpha}\bracket{\alpha}{t_j}$. So we obtain
\begin{eqnarray}
 \rho_A^T = \rho_A^* &=& \sum_{ij}\sum_\alpha \lambda_\alpha^* \bracket{t_i}{\alpha}^*\bracket{\alpha}{t_j}^*\ketbra{t_i}{t_j}\\
 &=& \sum_{ij}\sum_\alpha \lambda_\alpha\bracket{\alpha}{t_i}\bracket{t_j}{\alpha}\ketbra{t_i}{t_j}\\
 &=&\sum_\alpha \lambda_\alpha \left(\sum_i\bracket{\alpha}{t_i}\ket{t_i}\right)\left(\sum_i\bracket{\alpha}{t_i}\ket{t_i}\right)^\dagger,
\end{eqnarray}
where we used the fact that eigenvalues $\lambda_\alpha$ are real for the second line.
Therefore, $\{\ket{\bar{\alpha}}\}$ with $\ket{\bar{\alpha}}\equiv\sum_i\bracket{\alpha}{t_i}\ket{t_i}$ forms the eigenbasis of $\rho_A^T$, and hence $\{U\ket{\bar{\alpha}}\}$ is the eigenbasis of $\mathrm{tr}_B([\bar{u}_A\otimes I_B](\rho_{AB}))$.
So starting from Eq.~(\ref{ubar}), we obtain
\begin{eqnarray}
\pi_A \circ [\bar{u}_A\otimes I_B](\rho_{AB})&=&(1-\gamma)\sum_\alpha\sum_{ij}\sum_{kl} q_{ijkl} \bracket{\bar{\alpha}}{t_j}\bracket{t_i}{\bar{\alpha}} U\ketbra{\bar{\alpha}}{\bar{\alpha}}U^{\dagger}\otimes \ketbra{r_k}{r_l}+\gamma\frac{I}{d_A}\otimes \rho_B\\
&=&(1-\gamma)\sum_\alpha\sum_{ijmn}\sum_{kl} q_{ijkl} \bracket{\bar{\alpha}}{t_j}\bracket{t_i}{\bar{\alpha}} \bracket{t_n}{\bar{\alpha}}\bracket{\bar{\alpha}}{t_m}U\ketbra{t_n}{t_m}U^{\dagger}\otimes \ketbra{r_k}{r_l}+\gamma\frac{I}{d_A}\otimes \rho_B\\
&=&(1-\gamma)\sum_\alpha\sum_{ijmn}\sum_{kl} q_{ijkl} \bracket{t_j}{\alpha}\bracket{\alpha}{t_i} \bracket{\alpha}{t_n}\bracket{t_m}{\alpha}U\ketbra{t_n}{t_m}U^{\dagger}\otimes \ketbra{r_k}{r_l}+\gamma\frac{I}{d_A}\otimes \rho_B,
\end{eqnarray}
where we used $\bracket{t_i}{\bar\alpha} = \bra{t_i}(\sum_j\bracket{\alpha}{t_j}\ket{t_j})  = \bracket{\alpha}{t_i}$ for the third line. By comparing to Eq.~(\ref{upi}), we conclude that $\pi_A \circ [\bar{u}_A\otimes I_B](\rho_{AB}) = [\bar{u}_A\otimes I_B] \circ \pi_A(\rho_{AB})$, so $\bar{u}\in X_A(\pi_A)$.

\end{proof}

The complete result for qudits then follows, \change{again by combining with the fact that local semiclassical channels always output classical-quantum states}:
\begin{cor}
For $d_A>2$, diagonal discord is monotonically nonincreasing under any local discord nongenerating channel.
\end{cor}

Fig.~\fig{cp} summarizes the structure of different classes of free local operations in the theory of $\pi$.

\section{Continuity}
As mentioned, Refs.~\cite{criteria,wpm} brought up examples of states with maximally mixed marginals, where diagonal discord can be discontinuous. The discontinuity essentially comes from the maximization within the degenerate subspace: one can perturb the state in the direction that is far away from the optimal eigenbasis. However, in the absence of degeneracies, the eigenbasis is unique, so the above phenomenon cannot occur.
We first formally prove that diagonal discord is indeed continuous when the local density operator being measured is nondegenerate, by deriving a continuity bound \change{in a similar spirit as the celebrated Fannes-type inequalities for the continuity of the von Neumann entropy \cite{Fannes1973,audenaert}}.  The main idea is that $\pi$ changes continuously, which is also known as ``weak continuity'' \cite{criteria}. 

\begin{thm}\label{continuity}
Diagonal discord is continuous at states such that the local density operator being measured is nondegenerate. More explicitly,
let $\rho_{AB}$ be a bipartite state in finite dimensions such that $\rho_A={\rm tr}_B\rho_{AB}$ has distinct eigenvalues, and the smallest gap is $\Delta$. 
Suppose $\rho'_{AB}$ is a perturbed state such that $\norm{\rho'_{AB}-\rho_{AB}}_1\leq\epsilon$. For sufficiently small \change{$\epsilon>0$}, it holds that 
\begin{equation}
 |\bar{D}_A(\rho'_{AB})-\bar{D}_A(\rho_{AB})|\leq \left(\frac{\change{ \sqrt{2\,d_A^3d_B^3}}}{\Delta}+1\right)\epsilon\log (d_A d_B-1)+ H\left[\frac12\left(\frac{\change{2 \sqrt{2\,d_A^3d_B^3}}}{\Delta}+1\right)\epsilon\right]+ H(\epsilon/2).\label{eq:continuity_bound}
\end{equation}
where $H(\epsilon)=-\epsilon \log \epsilon -(1-\epsilon) \log (1-\epsilon) $ is the binary entropy function.
\end{thm}
\begin{proof}
In the following, we adopt matrix norms given by vectorization, i.e., for an operator $M$, $\norm{M}_p:=\norm{\mathrm{vec}(M)}_p$. For density matrices, $\norm{\cdot}_p$ is equivalent to the Schatten $p$-norm. In particular, $p=1$ yields the trace norm, and $p=2$ yields the Frobenius norm, also known as Hilbert-Schmidt norm or Schur norm.

Notice that 
\begin{eqnarray}
\left|\bar{D}_A(\rho'_{AB})-\bar{D}_A(\rho_{AB})\right|&=&\left|[S(\pi_A(\rho'_{AB}))-S(\rho'_{AB})]-[S(\pi_A(\rho_{AB}))-S(\rho_{AB})]\right|\\
&\leq&\left|S(\pi_A(\rho'_{AB}))-S(\pi_A(\rho_{AB}))\right|+\left|S(\rho_{AB})-S(\rho'_{AB})\right|,\label{ss}
\end{eqnarray}
where the inequality follows from the triangle inequality.
So, by the continuity of von Neumann entropy, diagonal discord is continuous as long as the discord-destroyed state $\pi_A(\rho_{AB})$ is continuous, that is, $\pi_A(\rho'_{AB})$ and $\pi_A(\rho_{AB})$ remain close. We show that it is so when $\rho_A$ is nondegenerate. (Indeed, discontinuity can occur in the vicinity of degeneracies, since the local eigenbases of perturbed states can be far from one another due to the nonuniqueness of eigenbases within the degenerate subspace, and hence $\pi_A(\rho_{AB})$ cannot be continuous. This is the essence behind the examples of discontinuities given in \cite{wpm,criteria}.)
Given $\rho_{A}=\sum_i p_i\Pi_i$, the spectral decomposition of the perturbed marginal can take the form $\rho'_{A}=\sum_i p'_i\Pi'_i$ with perturbed eigenvalues and eigenvectors, since they change continuously \cite{Kato:101545}.
By triangle inequality,
\begin{equation}
\norm{\pi_A(\rho'_{AB})-\pi_A(\rho_{AB})}_1\leq\norm{\pi_A(\rho'_{AB})-\sum_i(\Pi_i \otimes I)\rho'_{AB}(\Pi_i \otimes I)}_1+\norm{\sum_i(\Pi_i \otimes I)\rho'_{AB}(\Pi_i \otimes I)-\pi_A(\rho_{AB})}_1.
\label{tri}
\end{equation}
Since the trace distance is contractive \cite{nc}, the second term directly satisfies
    $\norm{\sum_i(\Pi_i \otimes I)\rho'_{AB}(\Pi_i \otimes I)-\pi_A(\rho_{AB})}_1\leq\norm{\rho'_{AB}-\rho_{AB}}_1\leq\epsilon$. 
The first term is also well bounded due to the continuity of eigenprojection $\pi_A$ \cite{Kato:101545}.
Now we derive an explicit bound for the first term. We assume that $\epsilon$ is sufficiently small so that any $\rho'_{A}$ still remains nondegenerate. (This is always possible since the spectrum is bounded away from a degenerate one by assumption.)
By triangle inequality,
\begin{eqnarray}
&&\norm{\pi_A(\rho'_{AB})-\sum_i(\Pi_i \otimes I)\rho'_{AB}(\Pi_i \otimes I)}_1\nonumber\\&=&\norm{\sum_i(\Pi'_i \otimes I)\rho'_{AB}(\Pi'_i \otimes I)-\sum_i(\Pi_i \otimes I)\rho'_{AB}(\Pi_i \otimes I)}_1\\
&\leq&\sum_i\left(\norm{(\Pi'_i \otimes I)\rho'_{AB}(\Pi'_i \otimes I)-(\Pi'_i \otimes I)\rho'_{AB}(\Pi_i \otimes I)}_1+\norm{(\Pi'_i \otimes I)\rho'_{AB}(\Pi_i \otimes I)-(\Pi_i \otimes I)\rho'_{AB}(\Pi_i \otimes I)}_1\right).
\end{eqnarray}
Notice that
\begin{eqnarray}
\norm{(\Pi'_i \otimes I)\rho'_{AB}(\Pi'_i \otimes I)-(\Pi'_i \otimes I)\rho'_{AB}(\Pi_i \otimes I)}_1 &=& \norm{(\Pi'_i \otimes I)\rho'_{AB}[(\Pi'_i-\Pi_i) \otimes I]}_1\\
&\leq& \sqrt{d_A d_B}\norm{(\Pi'_i \otimes I)\rho'_{AB}[(\Pi'_i-\Pi_i) \otimes I]}_2\\
&\leq&\sqrt{d_A d_B}\norm{(\Pi'_i \otimes I)}_2\norm{\rho'_{AB}}_2\norm{(\Pi'_i-\Pi_i) \otimes I}_2\\
&\leq&d_B\sqrt{d_A d_B}\norm{\Pi'_i-\Pi_i}_2,
\end{eqnarray}
where the second line follows from $\norm{M}_1\leq\sqrt{\mathrm{rank}M}\norm{M}_2$ \cite{srebro}, the third line follows from submultiplicativity of the Frobenius norm, and the fourth line follows from $\norm{\cdot}_2\leq\norm{\cdot}_1$ \cite{srebro} and $\norm{\rho'_{AB}}_1=1$. Similarly for the second term. So, we obtain
\begin{equation}\label{pi}
\norm{\pi_A(\rho'_{AB})-\sum_i(\Pi_i \otimes I)\rho'_{AB}(\Pi_i \otimes I)}_1\leq 2d_B\sqrt{d_A d_B}\sum_i \norm{\Pi'_i-\Pi_i}_2.
\end{equation}

\change{We next derive an upper bound for $\norm{\Pi_i'-\Pi_i}_2$. 
Let $\xi\tau_{A}:=\rho_A'-\rho_A$ where $\tau_A$ is a traceless Hermitian operator with $\norm{\tau_A}_2=1$ and $\xi\geq 0$ is a scaling constant.
We have
\ba
\xi=\norm{\xi\tau_A}_2=\norm{\rho_A'-\rho_A}_2\leq \norm{\rho_A'-\rho_A}_1\leq \norm{\rho_{AB}'-\rho_{AB}}_1\leq \epsilon,
\label{eq:xi_tau_bound}
\ea
where the first inequality follows from $\norm{\cdot}_2\leq\norm{\cdot}_1$, and the second inequality follows from the contractivity of the trace norm.
Now, notice that $\norm{\Pi'_i-\Pi_i}_2=\sqrt{2(1-|\bracket{i'}{i}|^2)}$.
By nondegenerate perturbation theory \cite{sakurai_napolitano_2017,Kato:101545}, we express $\ket{i'}$ as $\ket{i'}={Z_i}^{\frac12}\ket{\tilde{i'}}$, where $\ket{\tilde{i'}}$ is the unnormalized perturbed state $\ket{\tilde{i'}}=\ket{i}+\xi\ket{i^{(1)}}+\xi^2\ket{i^{(2)}}+\dots$ with $\ket{i^{(k)}}$ being the $k$-th order correction, and ${Z_i}$ is the normalization constant. $Z_i$ has the form
\ba
 Z_i^{-1}=\bracket{\tilde{i'}}{\tilde{i'}}=1+\xi^2\sum_{j\neq i}\frac{|\bra{j}\tau_A\ket{i}|^2}{(\lambda_i-\lambda_j)^2}+\mathcal{O}(\xi^3),
\ea
where $\lambda_i$ denote the eigenvalues of $\rho_A$.
Since $\lambda_i-\lambda_j\geq \Delta>0$ for all $i,j$ with $i\neq j$ by assumption (recall that $\Delta$ is a constant determined by the spectrum of $\rho_A$), this perturbation series converges for sufficiently small $\xi$. 
Since $Z_i^{\frac12}=\bracket{i}{i'}$ because of the structure of the perturbation series \cite{sakurai_napolitano_2017}, for sufficiently small $\xi$, we have 
\ba
 1-|\bracket{i}{i'}|^2=1-Z_i=\xi^2\sum_{j\neq i}\frac{|\bra{j}\tau_A\ket{i}|^2}{(\lambda_i-\lambda_j)^2}+\mathcal{O}(\xi^3)\leq 2\xi^2\sum_{j\neq i}\frac{|\bra{j}\tau_A\ket{i}|^2}{(\lambda_i-\lambda_j)^2},
 \label{eq:series_bound}
\ea
where in the inequality we used that the higher-order terms approach zero more rapidly than the second-order term with $\xi\to 0$, so they can be bounded by $\xi^2\sum_{j\neq i}\frac{|\bra{j}\tau_A\ket{i}|^2}{(\lambda_i-\lambda_j)^2}$ for sufficiently small $\xi$.
Eq.\ \eqref{eq:series_bound} is guaranteed to hold for sufficiently small $\epsilon$ as well, since $\xi$ is bounded from above by $\xi\leq \epsilon$ due to Eq.\ \eqref{eq:xi_tau_bound}. 
Therefore, it holds for sufficiently small $\epsilon$ that
\ba
 \norm{\Pi_i'-\Pi_i}_2=\sqrt{2(1-|\bracket{i'}{i}|^2)}\leq \sqrt{2}\xi\sqrt{2\sum_{j\neq i}\frac{|\bra{j}\tau_A\ket{i}|^2}{(\lambda_i-\lambda_j)^2}}\leq \sqrt{2}\xi\frac{\sqrt{\sum_{i,j}|\bra{j}\tau_A\ket{i}|^2}}{\Delta}\leq \frac{\sqrt{2}\epsilon}{\Delta}.
\ea
where the last inequality is due to Eq.\ \eqref{eq:xi_tau_bound} and $\norm{\tau_A}_2=\sqrt{\sum_{i,j}|\bra{j}\tau_A\ket{i}|^2}=1$.
}

 Plugging this result into Eq.\ (\ref{pi}) and then Eq.\ (\ref{tri}), we get
\begin{equation}
\norm{\pi_A(\rho'_{AB})-\pi_A(\rho_{AB})}_1 \leq \change{\left(\frac{2\sqrt{2\,d_A^3d_B^3}}{\Delta}+1\right)\epsilon
}.
\label{eq:pi_final}
\end{equation}
By the Fannes-Audenaert inequality \cite{Fannes1973,audenaert},
\begin{eqnarray}
|S(\pi_A(\rho'_{AB}))-S(\pi_A(\rho_{AB}))|&\leq& \frac12\left(\frac{\change{2 \sqrt{2\,d_A^3d_B^3}}}{\Delta}+1\right)\epsilon\log (d_A d_B-1)+ H\left[\frac12\left(\frac{\change{2 \sqrt{2\,d_A^3d_B^3}}}{\Delta}+1\right)\epsilon\right],\\
|S(\rho_{AB})-S(\rho'_{AB})|&\leq& \frac{\epsilon}{2}\log (d_A d_B-1)+H(\epsilon/2),
\end{eqnarray}
where $H$ is the binary entropy function. By Eq.\ (\ref{ss}),
\begin{equation}
    |\bar{D}_A(\rho'_{AB})-\bar{D}_A(\rho_{AB})|\leq \left(\frac{\change{ \sqrt{2\,d_A^3d_B^3}}}{\Delta}+1\right)\epsilon\log (d_A d_B-1)+ H\left[\frac12\left(\frac{\change{2 \sqrt{2\,d_A^3d_B^3}}}{\Delta}+1\right)\epsilon\right]+ H(\epsilon/2).
\end{equation}

The source of discontinuity in the presence of degeneracies is essentially the first term of the right hand side of Eq.~(\ref{tri}): $\sum_i(\Pi_i \otimes I)\rho'_{AB}(\Pi_i \otimes I)$ is not necessarily close to $\pi_A(\rho'_{AB})$.  
\end{proof}
\change{
\begin{rem}
 Using Theorem \ref{continuity}, one can find an explicit form of $\epsilon$ to achieve a certain target accuracy $E>0$ for the diagonal discord, thereby obtaining an $\epsilon$-$\delta$ statement of continuity. 
 Since Eq.~(\ref{eq:continuity_bound}) is only applicable to sufficiently small $\epsilon$, we restrict our attention to the regime of sufficiently small $E$ and $\epsilon$, which is sufficient for the sake of demonstrating continuity.   
 Note that $\epsilon < \sqrt{\epsilon}$ and $H(\epsilon)< 2\sqrt{\epsilon}$ for $0<\epsilon<1$.  
 Writing $a=\left(\frac{\change{ \sqrt{2\,d_A^3d_B^3}}}{\Delta}+1\right)\log (d_A d_B-1)$, and $b= \sqrt{\frac{ \sqrt{2\,d_A^3d_B^3}}{\Delta}+\frac12}$, we get
 $|\bar{D}_A(\rho'_{AB})-\bar{D}_A(\rho_{AB})|< (a+2b+\sqrt{2})\sqrt{\epsilon}$.
 Therefore, to achieve $|\bar{D}_A(\rho'_{AB})-\bar{D}_A(\rho_{AB})| < E$, it is sufficient to take $\epsilon < \left(E/(a+2b+\sqrt{2})\right)^2$.  Note that the inequalities in the above $\epsilon$-$\delta$-criterion are strict inequalities.
\end{rem}
}
Locally nondegenerate states such that the local eigenbasis minimizes discord (and deficit), which we call $\pi$-optimal states, represent an important class of states such that the restriction to eigenbasis is indeed optimal. Note that all locally nondegenerate zero discord states are $\pi$-optimal states. The above continuity result indicates some special properties of $\pi$-optimal states.  For example, it directly follows from Theorem \ref{continuity} that diagonal discord remains close to optimized discords in the vicinity of $\pi$-optimal states. Also, continuity of the optimal basis (termed ``strong continuity'' \cite{criteria}) is known to fail for discord and deficit. However, we conjecture that strong continuity holds at $\pi$-optimal states.

\section{Generalizations}
The above results can be generalized to a wide variety of simple discord-type measures defined by $\pi$, such as different distances and multi-sided measures, which can be seen as close variants of diagonal discord. 

\subsection{Diagonal discord given by other distance measures}
First, consider general distance measures besides relative entropy. Let $\delta$ be a nonnegative real function satisfying $\delta(\rho,\sigma)=0$ iff $\rho=\sigma$. Consider
\begin{equation}
    \bar{\mathfrak{D}}(\rho_{AB})_{\delta,\pi_A}:= \delta(\rho_{AB},\pi_A(\rho_{AB}))
\end{equation} 
as a discord measure defined by $\delta$ and the resource destroying map $\pi_A$. 
If $\delta$ satisfies $\delta(\ee(\rho),\ee(\sigma))\leq \delta(\rho,\sigma)$, $\bar{\mathfrak{D}}(\rho_{AB})_{\delta,\pi_A}$ is monotonically nonincreasing under $X_A(\pi_A)$ \cite{rdm}:
\begin{cor}
 If $\delta$ is contractive, $\bar{\mathfrak{D}}(\rho_{AB})_{\delta,\pi_A}$ is monotonically nonincreasing under ${\bf SC \cup ISO}$ on $A$.
\end{cor}
Furthermore, the continuity holds when $\delta$ is given by the Schatten-$p$ norm:
\begin{thm}\label{g_continuity}
Let $\rho_{AB}$ be a bipartite state in finite dimensions such that $\rho_B={\rm tr}_A\rho_{AB}$ has distinct eigenvalues, and the smallest gap is $\Delta$. 
Suppose $\norm{\rho'_{AB}-\rho_{AB}}_1\leq\epsilon$ where $\epsilon$ is sufficiently small, it holds that 
\begin{equation}
\left|\bar{\mathfrak{D}}(\rho'_{AB})_{\norm{\cdot}_p,\pi_A}-\bar{\mathfrak{D}}(\rho_{AB})_{\norm{\cdot}_p,\pi_A}\right|\leq \change{2\left(1+\frac{\sqrt{2 d_A^3 d_B^3}}{\Delta}\right)}\epsilon.
\end{equation}
\end{thm}
\begin{proof}
By definition,
\begin{eqnarray}
 \left|\bar{\mathfrak{D}}(\rho'_{AB})_{\norm{\cdot}_p,\pi_A}-\bar{\mathfrak{D}}(\rho_{AB})_{\norm{\cdot}_p,\pi_A}\right|&=&\left|\norm{\rho'_{AB}-\pi_A(\rho'_{AB})}_p-\norm{\rho_{AB}-\pi_A(\rho_{AB})}_p\right|\\
 &\leq&\norm{\rho'_{AB}-\rho_{AB}-\left(\pi_A(\rho'_{AB})-\pi_A(\rho_{AB})\right)}_p\\
  &\leq&\norm{\rho'_{AB}-\rho_{AB}}_p+\norm{\pi_A(\rho'_{AB})-\pi_A(\rho_{AB})}_p\\
  &\leq&\norm{\rho'_{AB}-\rho_{AB}}_1+\norm{\pi_A(\rho'_{AB})-\pi_A(\rho_{AB})}_1\\
  &\leq&\change{2\left(1+\frac{\sqrt{2 d_A^3 d_B^3}}{\Delta}\right)}\epsilon,
\end{eqnarray}
where the first and the second inequalities follow from the triangle inequality, the third inequality follows from the monotonicity of Schatten norms $\norm{\cdot}_p\leq \norm{\cdot}_{p'}$ for $p\geq p'$, and the last inequality follows from the perturbation assumption and Eq.~\eq{pi_final}. Note that, as in Theorem \ref{continuity}, $\epsilon$ needs to be sufficiently small so that $\rho'_{A}$ always remains nondegenerate and the perturbation series converges.
\end{proof}
\change{An $\epsilon$-$\delta$ statement can be obtained in a similar manner as in the remark after Theorem \ref{continuity}.}

\subsection{Multi-sided diagonal discord}

In the above, we focused on the one-sided discord measures.
The results can be easily extended to multi-sided measures where we also make a measurement on system $\{A_k\}_{k=1}^n$ in such a way that it will not disturb the marginal state. Here, we assume that $n$ is finite. 
Let $\rho_{\{A_k\}}$ be a composite state over the systems $A_1,\dots, A_n$ and $\rho_{A_j}$ be nondegenerate for all $j=1\dots n$. Denote $\pi_{\{A_k\}}(\rho_{\{A_k\}})=\sum_{i_1\dots i_n}\left(\otimes_{k=1}^n \Pi_{i_k}\right)\rho_{AB}\left(\otimes_{k=1}^n \Pi_{i_k}\right)$ where $\{\ket{i_k}\}$ is the local eigenbasis of system $A_k$. Then we obtain the following.
\begin{cor}
 $\bar{\mathfrak{D}}(\rho_{\{A_k\}})_{\delta,\pi_{\{A_k\}}}$ is monotonically nonincreasing under local operations in ${\bf SC \cup ISO}$.
\end{cor}
\begin{cor}
 $\bar{\mathfrak{D}}(\rho_{\{A_k\}})_{\delta,\pi_{\{A_k\}}}$, where $\delta$ is Schatten-$p$ norm or relative entropy, is continuous at states such that the local density operators being measured are nondegenerate.
\end{cor}
We note that the known discord-type quantities given by local measurement in the eigenbasis belong to such generalizations when the local density operators being measured are nondegenerate.
$\bar{\mathfrak{D}}(\rho_{AB})_{S,\pi_A\pi_B}$ on a bipartite state (where $S$ denotes relative entropy) gives the measurement-induced disturbance \cite{mid}, and $\bar{\mathfrak{D}}(\rho_{AB})_{\norm{\cdot}_2,\pi_A}$ gives the measurement-induced nonlocality \cite{min} (the similar quantity given by geometric distance measure is investigated in \cite{guo2012}).

\change{\section{Comparison with optimized quantum discord} \label{sec:comparison}
The faithfulness, monotonicity and continuity properties shown above indicate that the diagonal discord is a reasonable measure of quantum correlation, even though it is easily calculable due to the natural, simplified strategy for determining the local measurement, in contrast to the original quantum discord and many variants. 
Here, we intend to gain further insights into the relation between these two quantities by numerically comparing them for an important class of two-qubit states, the symmetric $X$-states.
Recall that the quantum discord introduced by Olliver and Zurek is defined similarly to Eq.~\eqref{discord} while the maximum is instead taken over all the  local von Neumann measurements. 
The symmetric $X$-states we consider are the two-qubit states whose density matrices have the form
\ba
\begin{pmatrix}\label{eq:sx}
 a & 0 & 0 & w \\
 0 & b & z & 0 \\
 0 & z & b & 0 \\
 w & 0 & 0 & d 
\end{pmatrix}
\label{eq:xstates}
\ea
where all the entries are real numbers. 
The states in this class are known to play an important role in non-Markovian dynamics \cite{Fanchini2010}, and they also work as good benchmarks for the comparison because (very approximately correct) analytical formula for the quantum discord is known for this class of states \cite{Fanchini2010, PhysRevA.88.014302} while the states with $X$-state structure can cover the whole spectrum of the discord measure \cite{Girolami2011}.

In Fig.~\ref{fig:dd_discord}, we show the comparison between quantum discord ($D_A(\rho_{AB})$) and diagonal discord ($\bar{D}_A(\rho_{AB})$) for symmetric $X$-states randomly sampled from the geometry given by the generalized Bloch representation.  The justification and technical details of this sampling scheme is given in Appendix \ref{sec:x_sample}.  The point is that the states sampled according to this distribution can be regarded as reasonably random (although there is no naturally distinguished uniform measure for mixed states).
Recall that the diagonal discord is always an upper bound for the quantum discord (which is confirmed in Fig.~\ref{fig:dd_discord}).
We also find that diagonal discord matches the optimized discord exactly for a significant fraction of the sampled states.  That is, the optimal measurement for discord is given by an local eigenbasis for such states.
In our numerical experiment of $10^4$ random samples, we find the fraction of such instances to be approximately 32\% (recall that this fraction is with respect to the distribution induced by the generalized Bloch representation; see Appendix \ref{sec:x_sample}).  This non-vanishing fraction highlights the special role of local eigenbases, as they typically only represent a zero-measure subset of the set of all local measurements.
One might be worried about the large deviation of diagonal discord from quantum discord observed for some instances in Fig.~\ref{fig:dd_discord}.   
However, we stress that, now that diagonal discord is shown to be a valid faithful measure as explained in the above sections, one should regard optimized discord and diagonal discord as the measures corresponding to two different ways of characterizing the quantum correlation, and which measure is preferable just depends on the physical or operational setting one is interested in (for instance, see \cite{heat,dddemon,sone2018quantifying,sone2018nonclassical} for several scenarios in which diagonal discord plays the major role).

Although more thorough investigation would be necessary to draw a definite conclusion on generic states, we expect that a similar behavior would still be observed because of the capability of the $X$-states to cover the broad range of spectrum.  

\begin{figure}[htbp]
    \centering
    \includegraphics[width=0.33\textwidth]{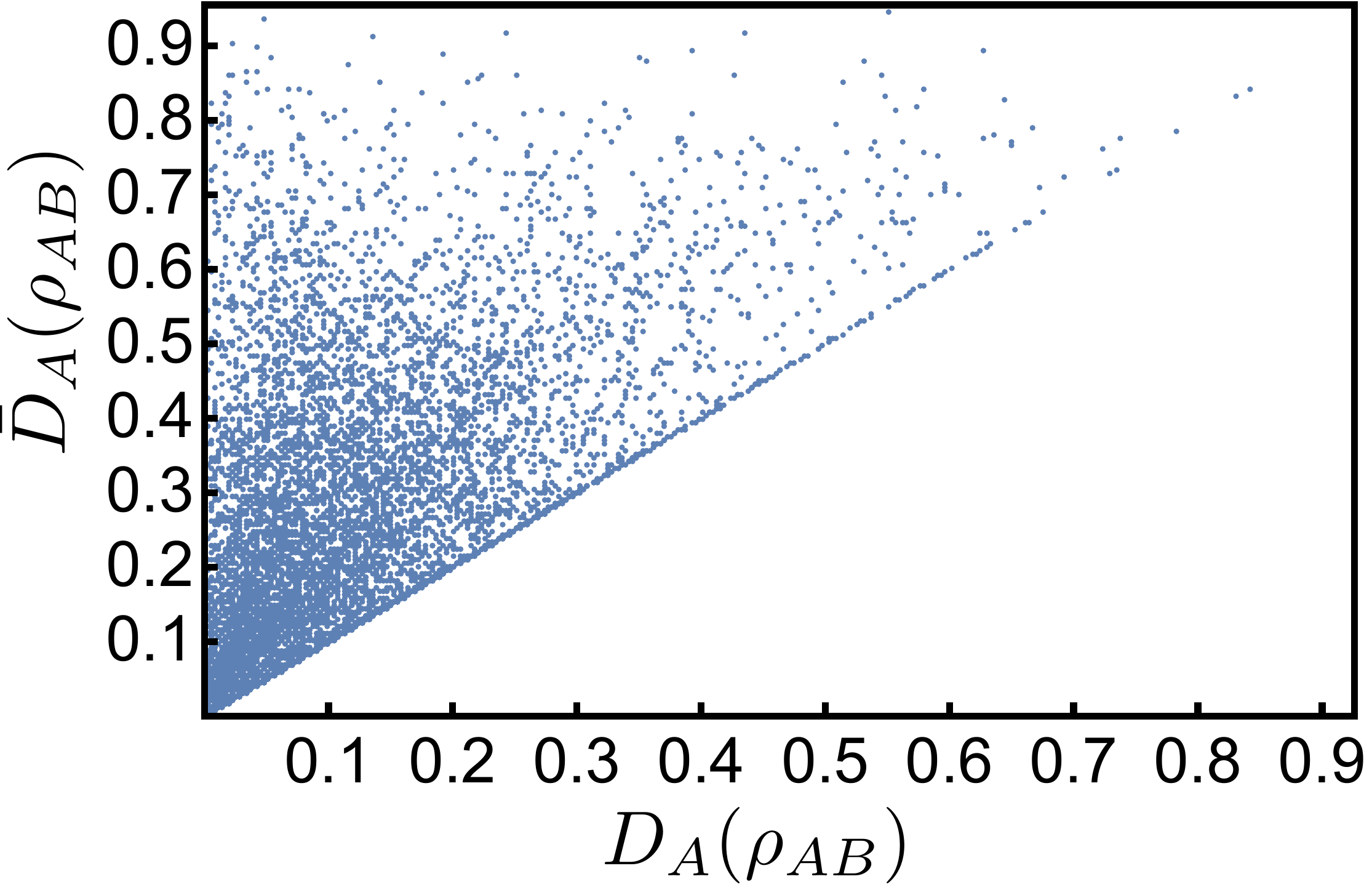}
    \caption{Quantum discord ($D_A(\rho_{AB})$) and diagonal discord ($\bar{D}_A(\rho_{AB})$) computed for the symmetric two-qubit $X$-states with the form of Eq.~\eqref{eq:xstates}, randomly sampled from the uniform distribution induced by the generalized Bloch representation. The number of samples is set to $10^4$.}
    \label{fig:dd_discord}
\end{figure}
}

\section{Concluding remarks}
Diagonal discord is an easily computable and natural measure of 
discord that has potentially wide application.    Here we showed that diagonal discord and a variety of similar measures exhibit desirable mathematical properties of monotonicity
and continuity in the generic case that the measured subsystem is nondegenerate. In particular, our analysis indicates the somewhat surprising
result that diagonal discord is a monotonone under all local discord nongenerating qudit channels, $d>2$, and is
very likely a monotone for discord nongenerating qubit channels as well.   This result represents a nontrivial application of the theory of resource destroying map.
Moreover, the direct
thermodynamic interpretations of diagonal discord \cite{heat,dddemon,sone2018quantifying} suggests
that diagonal discord may play a particularly important role in the resource theory of quantum correlation in general.

\begin{acknowledgements}
We thank Xueyuan Hu, Huangjun Zhu, Davide Girolami, and anonymous referees for helpful comments on the draft.
This work is supported by AFOSR, ARO, IARPA, and NSF under an INSPIRE
grant.
RT also acknowledges the support of the Takenaka Scholarship Foundation.
\end{acknowledgements}

\appendix

\change{
\section{Sampling symmetric $X$-states}
\label{sec:x_sample}
In Section \ref{sec:comparison}, we intend to compare the values of diagonal discord and ordinary (optimized) discord of some generic class of states.   In particular, we consider symmetric two-qubit $X$-states (which take the form of Eq.~(\ref{eq:sx})), since the (very approximately) correct analytical expression of optimized discord is known \cite{Fanchini2010,PhysRevA.88.014302}.  

To observe the generic behaviors, a scheme for randomly sampling symmetric two-qubit $X$-states states is needed. 
In particular, we need a distribution that is uniform in some sense to reasonably estimate the proportion of states such that the optimal basis for optimized discord is given by an eigenbasis, or equivalently, the optimized discord is exactly given by diagonal discord.
For mixed states, there is no unique, naturally distinguished uniform probability measure \cite{zyc_sommers1,bengtsson_zyczkowski_2006}.  Here, we use the following simple method. We express the two-qubit state in terms of the generalized Bloch representation \cite{PhysRevLett.47.838,KIMURA2003339,bengtsson_zyczkowski_2006,AERTS2014975}, and uniformly sample the allowed Bloch vector.  Such methods based on the Bloch representation is expected to give rise to a reasonable and natural notion of uniform distribution of mixed states: for example, it is known that uniform sampling from the qubit Bloch ball corresponds to the Hilbert-Schmidt measure, a standard distribution of mixed states induced by the Hilbert-Schmidt metric or partial tracing over the environment of equal size as the system \cite{bengtsson_zyczkowski_2006}.  For higher dimensions the intuition is similar.

The technical details of our scheme are given below.  
The generalized Bloch representation of a general two-qubit (4-dimensional) takes the following form:
\begin{equation}
    R(\vec{r}) = \frac{1}{4}(I_4 + \sqrt{6}\vec{r}\cdot\vec{\Lambda}),
\end{equation}
where $I_4$ is the identity matrix, $\vec{r} = \{r_i\}_{i= 1,...,15}, r_i\in[-1,1]$ is the generalized Bloch vector, and $\vec{\Lambda} = \{\Lambda_i\}_{i= 1,...,15}$, in analogy to Pauli matrices of SU(2) and Gell-Mann matrices of SU(3), are the 15 Hermitian, traceless generators of SU(4):
\begin{eqnarray*}
\begin{array}{ccc}
 \Lambda_1=\left(\begin{array}{cccc}0&1&0&0\\
1&0&0&0\\
0&0&0&0\\
0&0&0&0\end{array}\right),   & \;\;
 \Lambda_2=\left(\begin{array}{cccc}0&-i&0&0\\
i&0&0&0\\
0&0&0&0\\
0&0&0&0\end{array}\right),   & \;\;
 \Lambda_3=\left(\begin{array}{cccc}1&0&0&0\\
0&-1&0&0\\
0&0&0&0\\
0&0&0&0\end{array}\right),   
\end{array}
\\
\begin{array}{ccc}
 \Lambda_4=\left(\begin{array}{cccc}0&0&1&0\\
0&0&0&0\\
1&0&0&0\\
0&0&0&0\end{array}\right),   &\;\;
 \Lambda_5=\left(\begin{array}{cccc}0&0&-i&0\\
0&0&0&0\\
i&0&0&0\\
0&0&0&0\end{array}\right),   & \;\;
 \Lambda_6=\left(\begin{array}{cccc}0&0&0&0\\
0&0&1&0\\
0&1&0&0\\
0&0&0&0\end{array}\right),   
\end{array}
\\
\begin{array}{ccc}
 \Lambda_7=\left(\begin{array}{cccc}0&0&0&0\\
0&0&-i&0\\
0&i&0&0\\
0&0&0&0\end{array}\right),   &\;\;
 \Lambda_8=\frac{1}{\sqrt{3}}\left(\begin{array}{cccc}1&0&0&0\\
0&1&0&0\\
0&0&-2&0\\
0&0&0&0\end{array}\right),  &\;\;
 \Lambda_9=\left(\begin{array}{cccc}0&0&0&1\\
0&0&0&0\\
0&0&0&0\\
1&0&0&0\end{array}\right),   
\end{array}
\\
\begin{array}{ccc}
 \Lambda_{10}=\left(\begin{array}{cccc}0&0&0&-i\\
0&0&0&0\\
0&0&0&0\\
i&0&0&0\end{array}\right),   &\;\;
 \Lambda_{11}=\left(\begin{array}{cccc}0&0&0&0\\
0&0&0&1\\
0&0&0&0\\
0&1&0&0\end{array}\right),   &\;\;
 \Lambda_{12}=\left(\begin{array}{cccc}0&0&0&0\\
0&0&0&-i\\
0&0&0&0\\
0&i&0&0\end{array}\right),
\end{array}
\\
\begin{array}{ccc}
 \Lambda_{13}=\left(\begin{array}{cccc}0&0&0&0\\
0&0&0&0\\
0&0&0&1\\
0&0&1&0\end{array}\right),   & \;\;
 \Lambda_{14}=\left(\begin{array}{cccc}0&0&0&0\\
0&0&0&0\\
0&0&0&-i\\
0&0&i&0\end{array}\right),   &\;\;
 \Lambda_{15}=\frac{1}{\sqrt{6}}\left(\begin{array}{cccc}1&0&0&0\\
0&1&0&0\\
0&0&1&0\\
0&0&0&-3\end{array}\right).
\end{array}
\end{eqnarray*}
They satisfy the orthogonality relation $\mathrm{tr}(\Lambda_i\Lambda_j) = 2\delta_{ij}$ (and also the standard commutation relations and Jacobi identities).
The point is that they form a standard ``orthonormal'' basis of Hermitian matrices in dimension 4, in analogy to the unit basis vectors in Euclidean space.
Note that, in contrast to the basic Bloch representation for qubits, there exist matrices inside the unit ball of $\vec{r}$ with negative eigenvalues, i.e.~do not represent valid density operators, in higher dimensions \cite{AERTS2014975}. So we need to add the constraint of positive semidefiniteness to guarantee that the matrix is a density matrix.

The constraints enforced by the form of symmetric $X$-states are the following. 
First, several entries are restricted to be zero, which implies:
\begin{equation}
    r_1 = r_4 = r_{11} = r_{13} = 0.
\end{equation}
Second, the entries are real numbers, which implies:
\begin{equation}
    r_2 = r_5 = r_7 = r_{10} = r_{12} = r_{14} = 0.
\end{equation}
Finally, the $b$ entries imply that
\begin{equation}
    \frac{1}{4}(1 - \sqrt{6}r_3 + \sqrt{2}r_8 + r_{15}) = \frac{1}{4}(1 - 2\sqrt{2}r_8 + r_{15}),
\end{equation}
so
\begin{equation}
    r_3 = \sqrt{3}r_8.
\end{equation}
Therefore, the Bloch representation of symmetric two-qubit $X$-states take the following form, in terms of the four free parameters $r_6, r_8, r_9, r_{15}$:
\begin{equation}
    X(\vec{r}) = \frac{1}{4}\left(I_4 + \sqrt{6}\left(\sqrt{3}r_8\Lambda_3 + r_6\Lambda_6 + r_8\Lambda_8+ r_9\Lambda_9 + r_{15}\Lambda_{15}\right)\right),
\end{equation}
and the matrix form is
\begin{equation}
    X(\vec{r}) =  
\frac{1}{4}\begin{pmatrix}
 1 + 4\sqrt{2}r_8 + r_{15} & 0 & 0 &  \sqrt{6}r_9 \\
 0 & 1 - 2\sqrt{2}r_8 + r_{15} &  \sqrt{6}r_6 & 0 \\
 0 &  \sqrt{6}r_6 & 1 - 2\sqrt{2}r_8 + r_{15} & 0 \\
  \sqrt{6}r_9 & 0 & 0 & 1 - 3 r_{15} 
\end{pmatrix}.
\label{eq:xstates2}
\end{equation}
To sample such states uniformly according to the Bloch geometry, we draw $r_6, r_8, r_9, r_{15}$ from the uniform distribution on $[-1,1]$, and further require that $\norm{\vec{r}}_2 = r_6^2 + 4r_8^2 + r_9^2 + r_{15}^2 \leq 1$ (so that the data point is on or inside the generalized Bloch ball) and that $X(\vec{r})$ is positive semidefinite (so that the data point represents a valid density operator).
}


%

\end{document}